\documentclass{article}
\usepackage[utf8]{inputenc}
\usepackage{amsmath}
\usepackage{amsthm}
\usepackage{amssymb}
\usepackage{graphicx}
\usepackage[authoryear]{natbib}

\makeatletter

\DeclareTextSymbolDefault{\textquotedbl}{T1}

\theoremstyle{plain}
\newtheorem{thm}{\protect\theoremname}
\theoremstyle{plain}
\newtheorem{lem}{\protect\lemmaname}
\theoremstyle{plain}
\newtheorem{prop}{\protect\propositionname}
\theoremstyle{plain}
\newtheorem{cor}{\protect\corollaryname}

\usepackage{amsfonts}
\usepackage[onehalfspacing]{setspace}

\newtheorem{theorem}{Theorem}\newtheorem*{assumption}{Assumption}\newtheorem{claim}[theorem]{Claim}

\title{Monotone Comparative Statics in the Calvert-Wittman Model \thanks{%
The authors would like to thank Mar\'{\i}a Eugenia Boza, Michael Coppedge, Rafael DiTella, John Griffin, Robert Fishman, Ricardo Hausmann, B.J. Lee, Scott Mainwaring, Alfonso Miranda, Daniel Ortega, Benjamin Radcliff, Dani
Rodrik, John Roemer, Cameron Shelton, Jorge Vargas, Romain Wacziarg, Kathryn Vasilaky, the Associate Editor, an anonymous reviewer, and
participants at seminars at IESA, the University of Notre Dame, Cal Poly and
at conferences organized by the Public Choice Society, LACEA, the Game Theory
Society and the Econometric Society for their suggestions. \ Adolfo De Lima, Giancarlo Bravo and Reyes Rodr\'{\i}guez Acevedo
provided first-rate research assistance. \ All errors are our responsibility.%
}}
\author{Francisco Rodr\'{\i}guez\thanks{%
Fiscal Affairs Department, International Monetary Fund, 1900 Pennsylvania Ave NW, Washington, DC 20431. Email: frrodriguezc@gmail.com.}
 \and Eduardo Zambrano\thanks{%
Corresponding Author. Department of Economics, Orfalea College of Business,
Cal Poly, San Luis Obispo, CA 93407. Email: ezambran@calpoly.edu.}}
\date{December 2021}

\makeatother

\providecommand{\corollaryname}{Corollary}
\providecommand{\lemmaname}{Lemma}
\providecommand{\propositionname}{Proposition}
\providecommand{\theoremname}{Theorem}

\begin{document}
\maketitle
\begin{abstract}
In this paper, we show that when policy-motivated parties can commit to a particular platform during a uni-dimensional electoral contest
where valence issues do not arise there must be a positive association between the policies preferred by candidates and the policies adopted in expectation in the lowest and the highest equilibria of the electoral contest. We also show that this need not be so if the parties cannot commit to
a particular policy. The implication is that evidence of a negative relationship between enacted and preferred policies is suggestive of parties that hold positions from which they would like to move from yet are unable to do so. 

\end{abstract}



\thispagestyle{empty}

\clearpage
\section{Introduction}

The Downsian model of politics assumes that candidates can commit to keep their policy promises once they reach office. \ Their
ability to commit allows them to manipulate policy proposals so as to garner the fraction of votes that maximizes their probability of winning. Political
competition thus leads to convergence of proposed policies to the median voter's ideal point. \ A number of refinements
of this model have been proposed in the literature since Downs's 1957 contribution, many of which have attempted to reverse the problematic hypothesis of
complete convergence in policy proposals implied by Downsian competition.\protect\footnote{
Useful surveys include Mueller (2003), Hinich and Munger(1997), and Roemer (2001).} \ Until
the late nineties, most of this literature generally took as given the underlying assumption of a perfect capacity of candidates to make credible commitments.\protect\footnote{
See also Crain (2004), Chattopadhyay and Duflo (2004), Lee, Moretti and Butler (2004) and Groseclose (2001). For
a useful survey of the citizen-candidate model and its dynamic extensions see Duggan and Martinelli (2015).
}   

Besley and Coate (1997) and Osborne and Slivinsky (1996), however,
showed that some of the key results of the Downsian model no longer
hold in a model of citizen-candidates in which policymakers are not
bound to keep to their campaign promises. In particular, electoral
competition need not lead to full or even partial convergence in policy
platforms once candidates lose their ability to make credible promises.
\ Indeed, a multiplicity of equilibria become feasible, some of which
entail very extreme policies being proposed in equilibrium.

An extensive literature has developed in the past two decades addressing
the issue of how the citizen-candidate assumptions can be reconciled
with the intuition of the spatial competition model. These contributions
typically model repeated game interactions in which politicians who
deviate from their promises are punished in future elections and thus
gain an incentive to hold to their campaign promises. (Alesina, 1988;
Dixit, Grossman and Gul, 2000; Aragonès, Palfrey and Postlewaite,
2007; Panova, 2017). In some settings, politicians may decide to maintain
ambiguity about their preferences either because they do not know
the true preferences of the median voter (Glazer, 1990), wish to provide
a signal of their character or avoid reputational risks (Kartik and
McAfee, 2007; Kartik and van Weelden, 2019). Empirical tests of the
credibility hypothesis include comparisons of campaign promises and
legislative votes (Sulkin, 2009; Bidwell, Casey and Glennerster 2020),
assessments of the effect of term limits on observed policies (Besley
and Case, 1995, 2003; Ferraz and Finan, 2011) or testing for opportunistic
policy cycles (Alesina et al., 1997; Shi and Svennson, 2006). 

The intuition for our result is simple.  There are policy platforms that are so extreme that it makes no sense for a rational politician to adopt them.  This is because extreme positions can drive away so many voters to both make their proponents less likely to win an election and drive the probability-weighted policy further from their ideal point.  It follows that if we observe politicians adopting such platforms, it must reflect their inability to credibly commit to more moderate policy platforms.

To derive testable hypotheses from this intuition, we study the shape of the \emph{expected policy function}, which maps candidates' platforms into expected policies. We argue that candidates who can make credible commitments will never position themselves on the downwards-sloping segment of the expected policy function, where further moderation would lead expected policies closer to their ideal points. If we find candidates adopting platforms that fall in this segment, that is a good reason to conclude that they are constrained from further moderation by the inability to make credible promises.   This idea is conceptually like the notion that a profit maximizing monopolist would not produce in the decreasing region of its revenue function where reducing output would simultaneously increase its revenues and decrease its costs.

When politicians can make credible commitments, platforms are endogenous variables.  This makes it difficult to empirically evaluate hypotheses about the relationship between platforms and policies. To address this issue, we show that the \emph{equilibrium indirect expected policy function}, which maps candidate preferences into equilibrium expected policies, is also always increasing in the ideal policies of the candidates and can thus be used to investigate whether credibility problems can arise in practice, even in the presence of multiple equilibria. We illustrate how this result can be used to empirically evaluate credibility theories, for example by studying the correlation between changes in constituent or political leaders' preferences as measured by opinion surveys, and enacted policies.

The
rest of the paper is organized as follows. Section \ref{Gen} presents the main results of the paper in detail. Section
\ref{Con} concludes.

\section{\label{Gen}Setting}

The policy space is the interval $T=\left[0,1\right].$ Voters have
ideal policies represented by a point in $T$. When faced with two
policies to choose from, the voter chooses the policy that is closest
in distance to the voter's ideal policy.

Candidate preferences are described by a continuous real-valued payoff
function $u:T^{2}\rightarrow\mathbb{R}$: where, for each ideal policy
$t\in T,$ $u\left(x,t\right)$ is strictly concave in platform $x\in T$,
with $u\left(t,t\right)>u\left(x,t\right)$ for all $x\neq t$. There
are two candidates, $l$ and $r$ with ideal policies $0\leq t_{l}<t_{r}\leq1$
who respectively choose platforms $x_{l}$ and $x_{r}$.

Voters' ideal policies are distributed over the policy space $T$
according to a density which is unknown to the candidates. Because
of this uncertainty, the policy, $\mathbf{m},$ preferred by the median
voter is uncertain and the candidates form beliefs about \textbf{$\mathbf{m}$}
according to a continuous distribution $F$ with full support. Given
the profile of platforms $\left(x_{l},x_{r}\right)$ proposed by the
candidates, the probability of candidate $l$ winning the election
is given by: 
\[
P\left(x_{l},x_{r}\right)=\left\{ \begin{array}{ccc}
F\left(\frac{x_{l}+x_{r}}{2}\right) &  & \text{if \ensuremath{x_{l}<x_{r}}}\\
\frac{1}{2} &  & \text{if \ensuremath{x_{l}=x_{r}}}\\
1-F\left(\frac{x_{l}+x_{r}}{2}\right) &  & \text{if }x_{l}>x_{r}
\end{array}\right.
\]
with the probability of $r$ winning the election simply being $1-P\left(x_{l},x_{r}\right)$.

In what follows we sometimes make additional assumptions about the
preferences and beliefs of the candidates. We will make it explicit
when those additional assumptions are called for.

For $i=l,r$, let $U_{t_{i}}\left(x_{l},x_{r}\right)$ denote the
expected payoff function for candidate $i$ with ideal policy $t_{i}$,
that is, 
\[
U_{t_{i}}\left(x_{l},x_{r}\right)=P\left(x_{l},x_{r}\right)u\left(x_{l},t_{i}\right)+\left(1-P\left(x_{l},x_{r}\right)\right)u\left(x_{r},t_{i}\right).
\]

\begin{assumption}[Strict Single Crossing Property]
If $t\leq x<x^{\prime}<y<y^{\prime}\leq t^{\prime}$ we have that 
\[
U_{t}\left(x^{\prime},y\right)\geq U_{t}\left(x,y\right)\Rightarrow U_{t}\left(x^{\prime},y^{\prime}\right)>U_{t}\left(x,y^{\prime}\right)
\]
and

\[
U_{t^{\prime}}\left(x^{\prime},y\right)\geq U_{t^{\prime}}\left(x^{\prime},y^{\prime}\right)\Rightarrow U_{t^{\prime}}\left(x,y\right)>U_{t^{\prime}}\left(x,y^{\prime}\right).
\]
\end{assumption}

To motivate the Strict Single Crossing Property (SSCP) assumption, it helps to understand why candidate $i$ would want
to adopt a platform other than $t_{i}$. The answer is: in order to
decrease the chance that $i$'s opponent wins (which would force candidate
$i$ to endure an enacted policy that is far from $i$'s ideal policy,
$t_{i}$). According to SSCP, if it (weakly) pays for candidate $i$
to moderate their platform when the opponent's platform is `nearby',
it definitely pays for candidate $i$ to moderate their platform when
the opponent's platform is `far.' This is so because when the opponent's
platform is `far', it is more painful for candidate $i$ to lose the
election.

\begin{assumption}[Strict Log Supermodularity]
For every $t,t^{\prime},x,x^{\prime},y$$\in T$ with $t<t^{\prime}\leq x<x^{\prime}<y$
or $y<x<x^{\prime}\leq t<t^{\prime}$

\[
\frac{u\left(x^{\prime},t^{\prime}\right)-u\left(y,t^{\prime}\right)}{u\left(x,t^{\prime}\right)-u\left(y,t^{\prime}\right)}>\frac{u\left(x^{\prime},t\right)-u\left(y,t\right)}{u\left(x,t\right)-u\left(y,t\right)}.
\]
\end{assumption}
The Strict Log Supermodularity (SLS) assumption pertains the strict log supermodularity in $\left(x,t\right)$
of the payoff difference function, $u\left(x,t\right)-u\left(y,t\right)$
over the set of platforms uniformly to the left, or uniformly to the
right, of the platform chosen by the opponent. This says that the
relative change in the difference in payoff between winning and losing
for a candidate that follows a certain increase in their platform
is greater when the candidate's ideal policy is high than when the
candidate's ideal policy is low. Examples of functions $u$ that satisfy
SLS include commonly used functions in the literature such as the
\emph{quadratic}, $u\left(x,t\right)=-\left(x-t\right)^{2},$ the
\emph{exponential} $u\left(x,t\right)=-e^{\left(x-t\right)}+x$, and
their positive, affine transformations. See, e.g., Duggan and Martinelli
(2017).\footnote{For a different example of an application of log supermodularity in
models of politics see Ashworth and Bueno de Mesquita (2006).}

In what follows, these assumptions will be employed as in the literature
on supermodular games: Assumption SSCP will be used to show that the best responses of each candidate are increasing in the platform
chosen by their opponent, to show that the set of Nash equilibria
is non-empty, and to show that this set has a smallest and a largest
element. Assumption SLS in turn will be used to show that the best responses of each candidate are increasing in their respective ideal
policies. Together, both assumptions, and the general structure of
our model, imply that the lowest and highest equilibria are increasing
in the ideal policies of the candidates, and that therefore the equilibrium indirect expected policy functions associated with the lowest and highest equilibria are  increasing in these ideal policies as
well. The reader interested in learning more about these techniques
work can consult Amir (2005).

\subsection{A model with commitment}

In this model, as in Calvert (1985) and Wittman (1977), candidate
$l$ sets their platform $x_{l}$ to solve 
\[
\max_{x_{l}}U_{t_{l}}\left(x_{l},x_{r}\right),
\]
taking $x_{r}$ as given.

Candidate $r$ sets their platform $x_{r}$ to solve

\[
\max_{x_{r}}U_{t_{r}}\left(x_{l},x_{r}\right),
\]

taking $x_{l}$ as given.


In what follows we investigate the characteristics of the Nash equilibria
of the game described above.

\subsubsection{The best responses of the candidates and their properties}

Let $\varphi_{i}:T\rightrightarrows T$ be the best response correspondence
for candidate $i$ $=l,r$. 
\begin{lem}
\label{G1}The best response correspondence \textup{$\varphi_{i}$}
for candidate i with ideal policy $t_{i}$ and platform, $x$, chosen
by i's opponent has the following properties: 
\[
\left\{ \begin{array}{ccc}
\varphi_{t_{i}}\left(x\right) \subset (x, t_{i}] &  & \text{if \ensuremath{x<t_{i}}}\\
\varphi_{t_{i}}\left(x\right)=\{t_{i}\} &  & \text{if \ensuremath{x=t_{i}}}\\
\varphi_{t_{i}}\left(x\right) \subset [t_{i}, x)  &  & \text{if }x>t_{i}
\end{array}\right.
\]
\end{lem}

All proofs are in the Online Appendix.

The interpretation is that candidate $i$'s best responses are always
`sandwiched' between $t_{i}$ and the platform chosen by $i$'s opponent,
$x$. 

Let $\overline{\varphi}_{t_{l}}\left(x\right)$ and $\underline{\varphi}_{t_{l}}\left(x\right)$ be, respectively, the largest and smallest elements of $\varphi_{t_{l}}\left(x\right)$.
\begin{prop}
\label{G4}Assume that SSCP holds. Then if $t_{l}\leq x_{r}<x_{r}^{\prime}\leq t_{r}$,
then we have that $\underline{\varphi}_{t_{l}}\left(x_{r}^{\prime}\right)\geq\overline{\varphi}_{t_{l}}\left(x_{r}\right),$
and if $t_{r}\geq x_{l}^{\prime}>x_{l}\geq t_{l}$, then we have that $\underline{\varphi}_{t_{r}}\left(x_{l}^{\prime}\right)\geq\overline{\varphi}_{t_{r}}\left(x_{l}\right).$ 
\end{prop}

The implication is that every selection of the best response correspondence of each candidate
is non-decreasing in the platform of their opponent over the set
of policies in $\left[t_{l},t_{r}\right].$

\subsubsection{The Nash equilibria of the game and their properties}
\begin{prop}
\label{G4-1}Assume that SSCP holds. Then the set $E$
of Nash equilibria is non-empty and it has (coordinatewise) largest
and smallest elements $\left(\overline{x}_{l}^{*},\overline{x}_{r}^{*}\right)$
and $\left(\underline{x}_{l}^{*},\underline{x}_{r}^{*}\right)$. 
\end{prop}

\begin{lem}
\label{G2}In every equilibrium $(x_{l}^{\ast},x_{r}^{\ast})$, $t_{l}\leq x_{l}^{\ast}<x_{r}^{\ast}\leq t_{r}.$ 
\end{lem}

This is the usual `partial convergence' result one obtains in the
Calvert-Wittman model. See, e.g, Roemer (1997), section 4.

\subsubsection{Equilibrium Comparative Statics}
\begin{thm}
\label{inc}Assume that SSCP and SLS hold. Let $t_{l}<t_{l}^{\prime}<t_{r}<t_{r}^{\prime}$. Then
\begin{itemize}
    \item $\overline{x}_{l}^{*}\left(t_{l}^{\prime},t_{r}\right)\geq \overline{x}_{l}^{*}\left(t_{l},t_{r}\right)$ and $\overline{x}_{r}^{*}\left(t_{l},t_{r}^{\prime}\right)\geq \overline{x}_{r}^{*}\left(t_{l},t_{r}\right)$
    \item $\underline{x}_{l}^{*}\left(t_{l}^{\prime},t_{r}\right)\geq \underline{x}_{l}^{*}\left(t_{l},t_{r}\right)$ and $\underline{x}_{r}^{*}\left(t_{l},t_{r}^{\prime}\right)\geq \underline{x}_{r}^{*}\left(t_{l},t_{r}\right)$
\end{itemize}

\end{thm}

To show this result we first establish that the best responses of each candidate 
are increasing in $t_{l}$ and $t_{r}$. This is illustrated in Figure
\ref{fig:Slide4.png}, which is drawn in  $\left[t_{l},t_{r}\right]\times \left[t_{l},t_{r}\right]$ space under the assumption that $\left(t_{l}^{\prime},t_{r}^{\prime}\right)>\left(t_{l},t_{r}\right),$ and with the best response correspondences being single-valued. The dashed blue line represents $\varphi_{t_{l}^{'}}$, and it is to
the right of the solid blue line, which represents $\varphi_{t_{l}}$.
The dashed gray line represents $\varphi_{t_{r}^{'}}$
and is above the solid gray line, which represents $\varphi_{t_{r}}$. Figure \ref{fig:Slide4.png} also illustrates the content
of Theorem \ref{inc}: the smallest equilibria of the model parametrized
by $\left(t_{l},t_{r}\right)$ is smaller than the smallest equilibria
of the model parametrized by $\left(t_{l}^{\prime},t_{r}^{\prime}\right).$
Similarly for the largest equilibria of the models. Figure \ref{fig:Slide4.png}
makes it clear that comparison of the rest of the equilibria may not
even be meaningful, since the model parametrized by $\left(t_{l},t_{r}\right)$
has an ``intermediate''\ equilibrium but the model parametrized
by $\left(t_{l}^{\prime},t_{r}^{\prime}\right)$ does not.

\begin{figure}
\centering \includegraphics[scale=0.35]{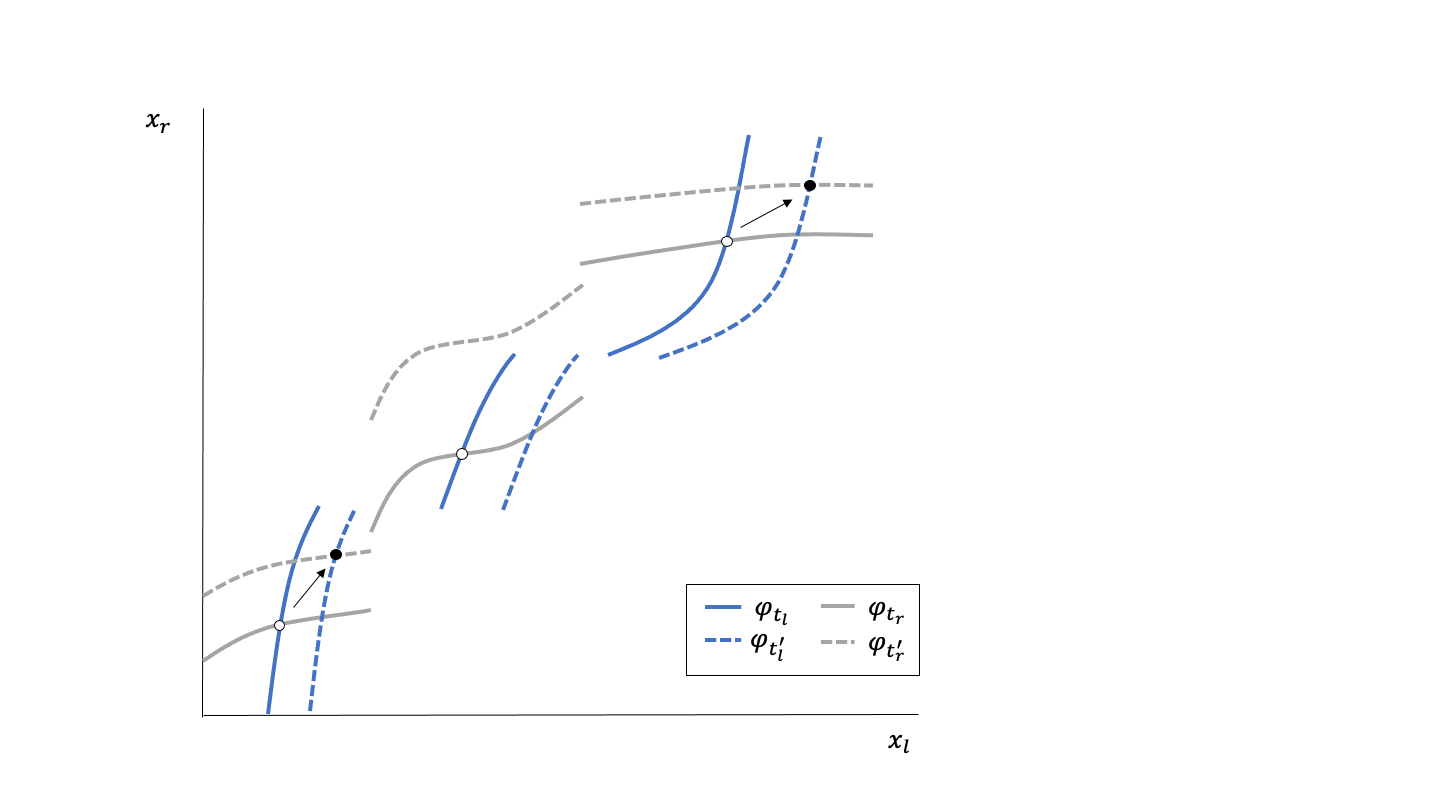} \caption{Multiple equilibria comparative statics}
\label{fig:Slide4.png} 
\end{figure}

Consider now the expected policy function, 
\[
\pi\left(x_{l},x_{r}\right):=P\left(x_{l},x_{r}\right)x_{l}+\left(1-P\left(x_{l},x_{r}\right)\right)x_{r}.
\]

The expected policy function estimates, before the resolution of uncertainty
about the electoral outcome, the platform that will ultimately be
adopted as policy. The left panel of Figure \ref{fig:Slide4.png-1}
plots the expected policy as a function of the platform, $x_{i},$
chosen by candidate $i$, given the platform, $x$, chosen by $i$'s
opponent. When $x_{i}=0$, $i$'s platform is too extreme to entail
a substantial probability of the candidate winning the election, and
the expected policy is therefore close to the platform chosen by $i$'s
opponent, $x$. As candidate $i$ moderates their platform, starting
from zero, $i$'s probability of winning increases, and the expected
policy therefore moves away from $x$. Eventually, as $x_{i}$ gets
close to $x$, so does the expected policy. Similarly, if $x_{i}=1$
and this platform is too extreme to entail any substantial probability
of candidate $i$ winning the election, the expected policy is close
to the platform chosen by $i$'s opponent, $x$. As candidate $i$
moderates their platform, starting from one, $i$'s probability of
winning grows, and the expected policy then begins to move away from
$x$. 
\begin{figure}
\centering

\includegraphics[scale=0.35]{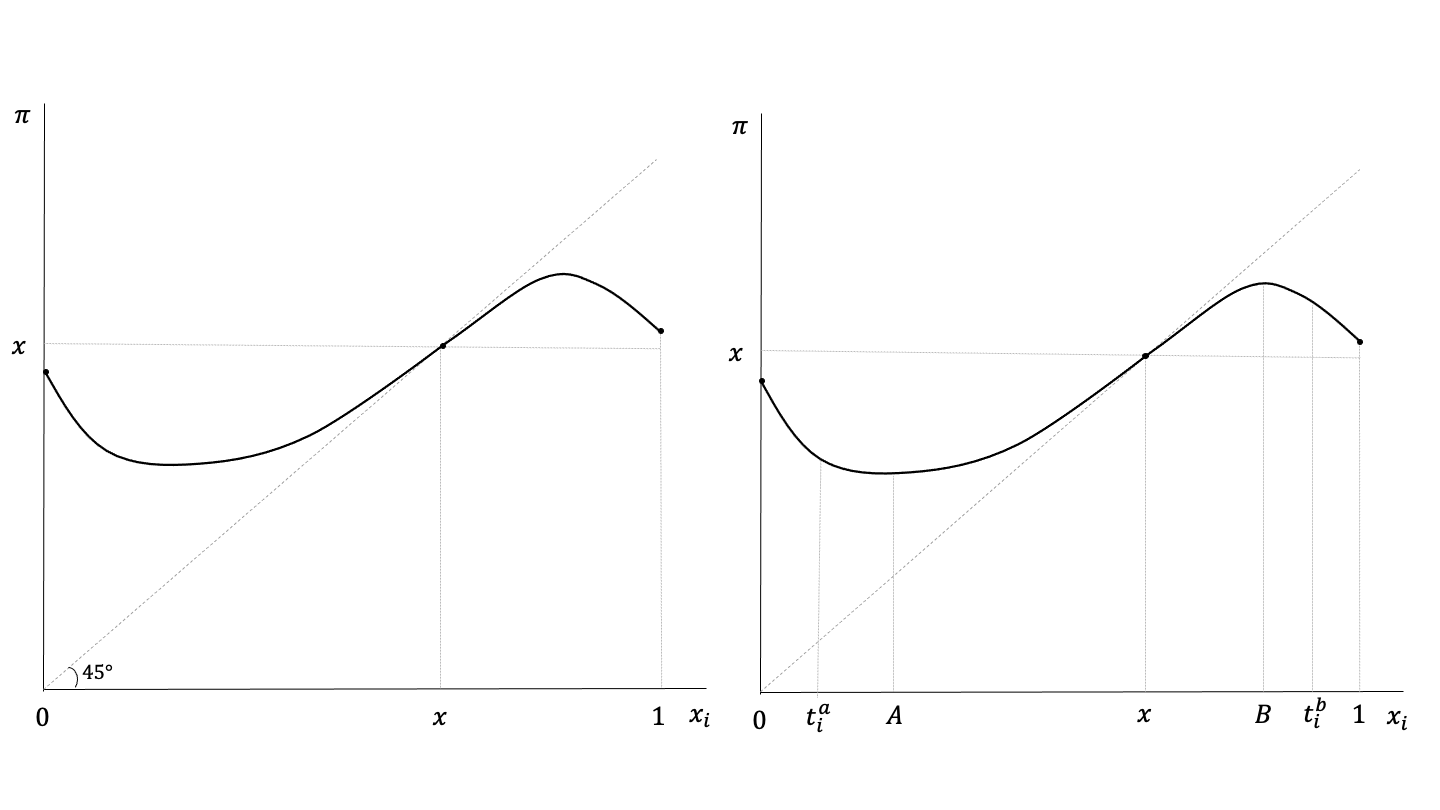} \caption{The Expected Policy Function}
\label{fig:Slide4.png-1} 
\end{figure}

\begin{thm}
\label{epo}Let $x_{r}>t_{l}$, \textup{$x_{l}\in\varphi_{t_{l}}\left(x_{r}\right)$
and $x_{l}^{\prime}>x_{l}$.} Then $\pi\left(x_{l}^{\prime},x_{r}\right)>\pi\left(x_{l},x_{r}\right).$
Let $x_{l}<t_{r}$, \textup{$x_{r}\in\varphi_{t_{r}}\left(x_{l}\right)$
and} \textup{$x_{r}^{\prime}<x_{r}$}. Then $\pi\left(x_{l},x_{r}^{\prime}\right)<\pi\left(x_{l},x_{r}\right).$ 
\end{thm}

Theorem \ref{epo} contains the main insight of the paper: a rational
candidate would select a platform that is in the increasing
region of the expected policy function. The right panel of Figure
\ref{fig:Slide4.png-1} illustrates this. If $t_{i}$ is in the increasing
region of the expected policy function, the result follows since Lemma
\ref{G1} shows that $\varphi_{t_{i}}\left(x\right)$is between $t_{i}$
and $x$. Now suppose that candidate $i$'s ideal policy is, say $t_{i}^{a}$
(resp. $t_{i}^{b}$). Then selecting a platform between $t_{i}^{a}$
and $A$ (resp. between $B$ and $t_{i}^{b}$) would leave unexploited
the possibility of increasing the expected payoff for the candidate
by moderating their platform, as this would drive the expected policy
closer to candidate $i$'s ideal point policy while at the same increasing
the candidate's probability of winning the election.

Since the platforms chosen by candidates in equilibrium are endogenous,
hypotheses testing that relies on direct estimation of the shape of
the expected policy function may be riddled with simultaneity bias.
In order to address this issue, we note that the \textit{equilibrium indirect
expected policy function, }which maps candidate preferences into expected policies for a given equilibrium, shares the same comparative
statics implications of the expected policy function and can thus
be used to investigate whether credibility problems can arise in practice,
even in the presence of multiple equilibria.

The \emph{equilibrium indirect expected policy function} can be computed as follows:

If $\left(x_{l}^{*},x_{r}^{*}\right)\in E$, then

\[
\pi^{*}\left(t_{l},t_{r};x_{l}^{*},x_{r}^{*}\right):=\pi\left(x_{l}^{*}\left(t_{l},t_{r}\right),x_{r}^{*}\left(t_{l},t_{r}\right)\right).
\]
Let $\overline{\pi}^{*}\left(t_{l},t_{r}\right)$ and $\underline{\pi}^{*}\left(t_{l},t_{r}\right)$
be the equilibrium indirect expected policy corresponding to the largest and smallest
equilibrium in $E$, respectively. That is, $\overline{\pi}^{*}\left(t_{l},t_{r}\right)=\pi^{*}\left(t_{l},t_{r};\overline{x}_{l}^{*},\overline{x}_{r}^{*}\right)$ and $\underline{\pi}^{*}\left(t_{l},t_{r}\right)=\pi^{*}\left(t_{l},t_{r};\underline{x}_{l}^{*},\underline{x}_{r}^{*}\right)$

We know from Theorem \ref{epo} that in equilibrium the expected policy
is increasing in $x_{l}$ and $x_{r}$. From Theorem \ref{inc} we
know that the largest and smallest Nash equilibria of the game, $\left(\overline{x}_{l},\overline{x}_{r}\right)$
and $\left(\underline{x}_{l},\underline{x}_{r}\right)$, are increasing
in $t_{l}$ and $t_{r}.$ It thus follows that the equilibrium indirect expected
policy functions associated with the largest and smallest equilibria
are also increasing in $t_{l}$ and $t_{r}$. This is the main comparative
statics result of the paper, which we summarize below. 
\begin{cor}
Assume that SSCP and SLS hold. If $t_{l}^{\prime}>t_{l}$ then $\overline{\pi}^{*}\left(t_{l}^{\prime},t_{r}\right)\geq\overline{\pi}^{*}\left(t_{l},t_{r}\right)$
and $\underline{\pi}^{*}\left(t_{l}^{\prime},t_{r}\right)\geq\underline{\pi}^{*}\left(t_{l},t_{r}\right).$
If $t_{r}^{\prime}<t_{r}$ then $\overline{\pi}^{*}\left(t_{l},t_{r}^{\prime}\right)\leq\overline{\pi}^{*}\left(t_{l},t_{r}\right)$
and $\underline{\pi}^{*}\left(t_{l},t_{r}^{\prime}\right)\leq\underline{\pi}^{*}\left(t_{l},t_{r}\right).$ 
\end{cor}

\subsection{A model without commitment}

When candidates cannot precommit to adopt a particular platform, voters
expect that, if elected, a candidate will implement their most preferred
policy once in office. Therefore, the candidates cannot affect the
probabilities of being elected and in the unique equilibrium, $x_{l}^{\ast}=t_{l}$
and $x_{r}^{\ast}=t_{r}.$\footnote{Because of politicians’ inability to make credible commitments, their expected payoffs are unaffected by the choice of platform and they therefore choose the platform that is closest to their ideal policy.} Because of this, the adopted platforms
are trivially increasing in $t_{l}$ and $t_{r}.$

It turns out, however, that in the model without commitment, Theorem
\ref{epo} fails and hence the indirect expected policy function \emph{need
not be} increasing in the ideal policies of the politicians, as in
the model with commitment. We illustrate that this is the case with
an example.

Consider a situation where candidates form beliefs about the policy
preferred by the median voter, $\mathbf{m},$ as follows: $\mathbf{m}$
is a random variable that is distributed according to a triangular
distribution in the {[}0,1{]} interval, with mode 0.5. We also let
$u\left(x,t\right)=-\left(x-t\right)^{2}$ with $x_{r}>0.5,$ although
nothing in the example depends on these choices.\footnote{The counterexample can be built with any probability distribution over $\text{\textbf{m}}$
such that $xf\left(x\right)>F\left(x\right)$ for some value of $x$.
Distributions with these characteristics abound and include, for example,
many instances from the Beta and Power families. The counterexample can also
be built using Roemer's error distribution model of uncertainty (Roemer
2001, section 2.3).} We then investigate the behavior of $P\left(x_{l},x_{r}\right)$
as $x_{l}$ varies given a fixed value of $x_{r}$, and of $\pi^{*}\left(t_{l},t_{r}\right)$
as $t_{l}$ varies given a fixed value of $t_{r}$. We obtain that
\[
P\left(x_{l},x_{r}\right)=\left\{ \begin{array}{ccc}
2\left(\frac{x_{l}+x_{r}}{2}\right)^{2} & \text{if} & \text{ \ensuremath{x_{l}\leq1-x_{r}}}\\
1-2\left(1-\frac{x_{l}+x_{r}}{2}\right)^{2} & \text{if} & \text{}1-x_{r}<x_{l}<x_{r}\\
\frac{1}{2} & \text{if} & \ensuremath{x_{l}=x_{r}}\\
2\left(1-\frac{x_{l}+x_{r}}{2}\right)^{2} & \text{if} & x_{l}>x_{r}
\end{array}.\right.
\]
\begin{figure}
\centering

\includegraphics[scale=0.3]{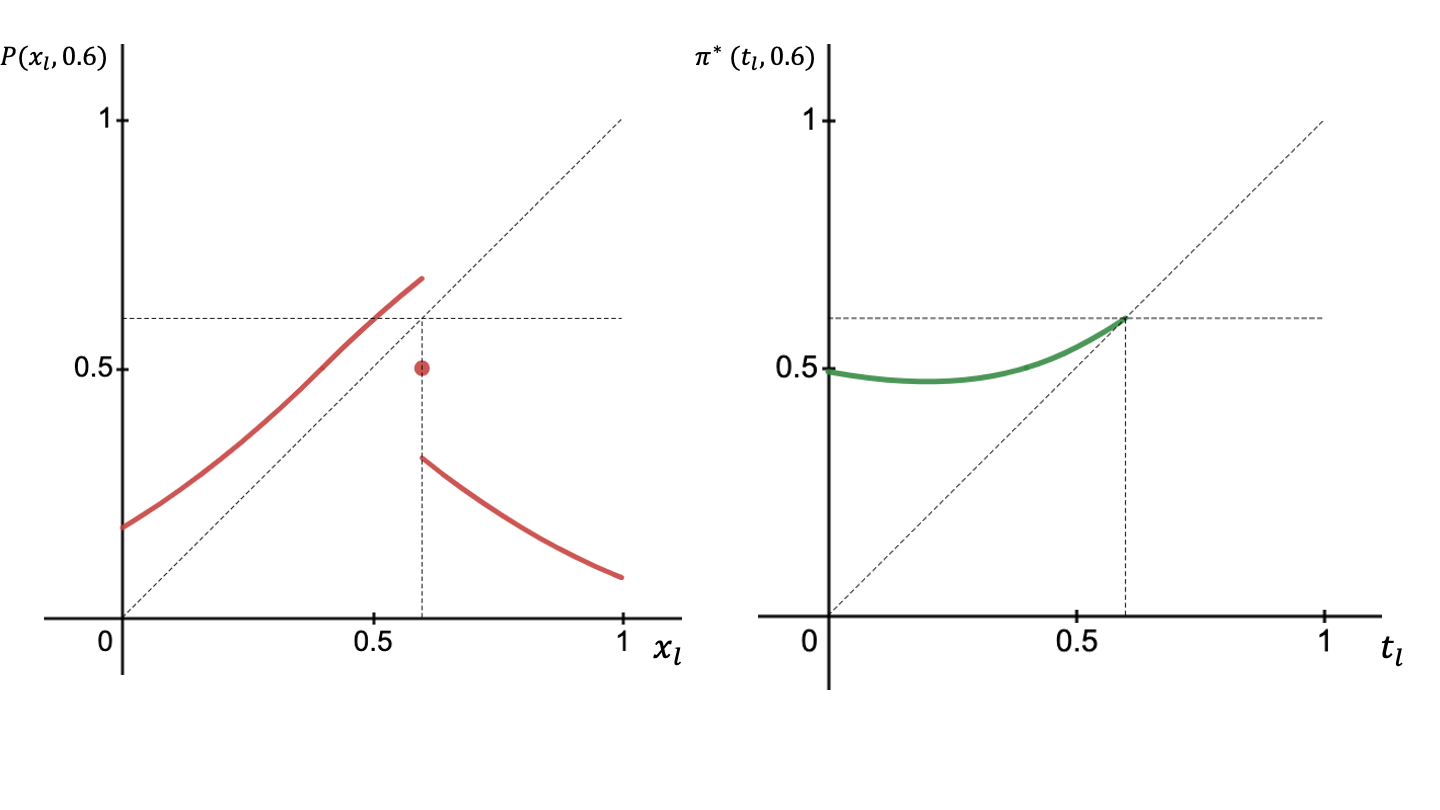} \caption{The Model Without Commitment}
\label{fig:Slide4.png-1-1} 
\end{figure}

The left panel of Figure \ref{fig:Slide4.png-1-1} represents the
behavior of $P\left(x_{l},x_{r}\right)$ given $x_{r}=0.6,$ and as
$x_{l}$ varies from zero to one. As expected, the probability of
candidate $l$ winning the election grows as the candidate's ideal
policy approaches $x_{r}=0.6$ from either side, and this probability
jumps to $0.5$ when both candidates have the same ideal policies.

The equilibrium indirect expected policy function in this case, when $x_{l}^{\ast}=t_{l}$
and $x_{r}^{\ast}=t_{r},$ can be more simply written as $\pi^{*}\left(t_{l},t_{r};t_{l},t_{r}\right)=\pi^{*}\left(t_{l},t_{r}\right)$, where 
\[
\pi^{*}\left(t_{l},t_{r}\right)=\left\{ \begin{array}{ccc}
2\left(\frac{t_{l}+t_{r}}{2}\right)^{2}\cdot t_{l}+\left[1-2\left(\frac{t_{l}+t_{r}}{2}\right)^{2}\right]\cdot t_{r} &  & \text{if \ensuremath{t_{l}\leq1-t_{r}}}\\
\left[1-2\left(1-\frac{t_{l}+t_{r}}{2}\right)^{2}\right]\cdot t_{l}+2\left(1-\frac{t_{l}+t_{r}}{2}\right)^{2}\cdot t_{r} &  & \text{if }1-t_{r}<t_{l}<t_{r}\\
t_{r} &  & \text{if \ensuremath{t_{l}=t_{r}}}\\
2\left(1-\frac{t_{l}+t_{r}}{2}\right)^{2}\cdot t_{l}+\left[1-2\left(1-\frac{t_{l}+t_{r}}{2}\right)^{2}\right]\cdot t_{r} &  & \text{if }t_{l}>t_{r}
\end{array},\right.
\]
which is a decreasing function of $t_{l}$ when evaluating the function
at any $t_{l}<\frac{t_{r}}{3}.$ To see this, notice that, when $t_{l}<1-t_{r},$
\[
\frac{d\pi^{*}\left(t_{l},t_{r}\right)}{dt_{l}}=\frac{1}{2}\left(3t_{l}^{2}+2t_{l}t_{r}-t_{r}^{2}\right).
\]
We obtain that $\frac{d\pi^{*}\left(0,t_{r}\right)}{dt_{l}}=-t_{r}^{2}<0,$and
$\frac{d^{2}\pi^{*}\left(t_{l},t_{r}\right)}{dt_{l}^{2}}=3t_{l}+t_{r}>0$.
Therefore, as $t_{l}$ grows from zero, $\frac{d\pi^{*}\left(t_{l},t_{r}\right)}{dt_{l}}$
becomes less negative, until it reaches zero, when $t_{l}=\frac{t_{r}}{3}$,
which is the only positive root of $3t_{l}^{2}+2t_{l}t_{r}-t_{r}^{2}$.

Hence, if the ideal point for candidate $l$ happens to be to the
left of $\frac{t_{r}}{3},$ the indirect expected policy function
will be decreasing in $t_{l}\ $ at that point. The right panel of
Figure \ref{fig:Slide4.png-1-1} represents the behavior of $\pi^{*}\left(t_{l},t_{r}\right)$
given $t_{r}=0.6,$ and as $t_{l}$ varies from zero to 0.6. The expected
policy drops as candidate $l$'s ideal policy approaches 0.2, as explained
above, and subsequently rises as candidate l's ideal policy grows
beyond 0.2, and all the way up to 0.6.

As this example shows, it is not hard to find cases in which candidates who cannot make credible commitments will have policy positions that fall on the downward-sloping segment of the indirect expected policy function.  This follows from the fact that without a commitment technology, policy platforms will simply reflect candidate preferences.  Some candidates have preferences that are so extreme that it would be in their interest to credibly commit to being more moderate if they could do so.  That they do not do so is thus good evidence of their inability to credibly make such promises.

This marks an important difference from the commitment case, in which candidates can and do make such promises.  In the presence of a commitment technology, extreme candidates will simply decide to moderate their policy platform to the level at which moderation drives the expected policy as close as possible to their ideal point. Therefore, extreme policy positions (in the precise sense of being so extreme that they drive expected policy away from the politician’s ideal point) are inconsistent with rational politicians being able to make credible commitments.

\section{Conclusions\label{Con}}

We have shown that when candidates can commit to a particular platform
during a uni-dimensional electoral contest where valence issues do
not arise there must be a positive association between the policies
we can expect will be adopted in (the smallest and the largest) equilibrium
and the preferred policies held by the candidates. We have also shown
that this need not be so if the candidates cannot commit to a particular
policy. The implication is that evidence of a negative relationship
between enacted and preferred policies in the data is suggestive of
candidates that hold positions from which they would like to move
from yet are unable to do so. This is the main result of the paper.

This approach can be extended to other models of policy location.
For example, Groseclose (2001) proposed a model in which a difference
in valence can lead candidates to assume extreme positions. Non-trivial
valence differences would violate our symmetry and -- under the Groseclose
conditions -- our monotonicity assumptions, so the approach taken
in Section 2 is not directly suitable for testing a valence model.
Future research could then focus on (i) allowing for valence and multidimensional
issues to play a role, and (ii) understanding what assumptions on
the beliefs held by the candidates about the distribution of voter
preferences, in lieu of SSCP and SLS, would allow our approach to equilibrium
existence and comparative statics to be applicable in these cases.

Our results suggest that empirical work on testing for the existence
of credibility problems in politics could advance through direct estimation
of the direct and indirect expected policy functions. A regression
of enacted policies on policy platforms could shed light on whether
the observed correlation between these is positive, as suggested by
models of commitment, or negative, as would be the case in citizen-candidate
environments. In order to address simultaneity problems in the estimation
of the expected policy function, platforms could be instrumented on
measures of policymaker or constituent preferences drawn from public
opinion surveys, in effect helping us recover the indirect expected
policy function.

Anectodally, examples of candidates whose platforms around a single
issue were simply too extreme for their own good abound (e.g., George
McGovern in 1972 against Richard Nixon and Mario Vargas Llosa in 1990
against Alberto Fujimori). A conventional analysis of the behavior
of these politicians would characterize the behavior as relying on
gross miscalculations, based on mistaken beliefs about what voters'
actual preferences really were. \ Under the alternative interpretation
that we espouse, there is nothing irrational about these policy platforms.
\ It wasn't the policy platforms of these politicians that cost them
the elections: it was their preferences. \ Had they proposed more
moderate platforms, voters would not have bought it. The presumption
that these politicians do not understand the political environment
in which they operate is not needed to explain how we see these politicians
behaving during election time.

\section{References}

Alesina, Alberto (1988), ``Credibility and Policy Convergence in
a Two-candidate System with Rational Voters,''\ American Economic
Review 78, 796-806.

Alesina, Alberto; Nouriel Roubini and Gerald Cohen (1997). Political
Cycles and the Macroeconomy. Berkeley: University of California Press.

Enriqueta Aragonès, Andrew Postlewaite, Thomas Palfrey (2007), Political
Reputations and Campaign Promises, Journal of the European Economic
Association 5(4), 846-884.

Amir, Rabah (2005), ``Supermodularity and complementarity in Economics:
An Elementary Survey,'' Southern Economic Journal 71 (3), 636-660.

Ashworth, Scott and Ethan Bueno de Mesquita (2006), ``Monotone Comparative
Statics and Models of Politics,''\ American Journal of Political
Science 50, 214-231.

Bagnoli, Mark and Ted Bergstrom (2005), ``Log-concave probability
and its applications,'' Economic Theory 26, 445-469.

Besley, Timothy and Anne Case (1995), “Does Electoral Accountability
Affect Economic Policy Choices? Evidence from Gubernatorial Term Limits.”
Quarterly Journal of Economics 110 (3), 769-798.

Besley, Timothy and Anne Case (2003), “Political Institutions and
Policy Choices: Evidence from the United States.” Journal of Economic
Literature 41(1), 7-73.

Besley, Timothy and Stephen Coate (1997), ``An Economic Model of
Representative Democracy,''\ Quarterly Journal of Economics 112,
85-114.

Bidwell, Kelly, Katherine Casey and Rachel Glennerster (2020), “Debates:
Voting and Expenditure Responses to Political Communication.” Journal
of Political Economy 128(8), 2880-2924.

Calvert, Randall (1985), ``Robustness of the Multidimensional Voting
Model: Candidate Motivations, Uncertainty, and Convergence,''\ American
Journal of Political Science 29, 69-95.

Crain, Mark. (2004), ``Dynamic Inconsistency,''\ in C. Rowley and
F. Schneider (eds.), Encyclopaedia of Public Choice, 481-484.

Chattopadhyay, Raghabendra and Esther Duflo (2004), ``Women as Policy
Makers: Evidence from a Randomized Policy Experiment in India,''\ Econometrica
72(5), 1409-1443.

Dixit, Avinash, et al. (2000) “The Dynamics of Political Compromise.”
Journal of Political Economy 108(3), 531-568.

Duggan, John and César Martinelli (2017), ``The Political Economy
of Dynamic Elections: Accountability, Commitment and Responsiveness,''
Journal of Economic Literature 55, 916-984.

Ferraz, Claudio, and Frederico Finan (2011), ``Electoral Accountability
and Corruption: Evidence from the Audits of Local Governments," American
Economic Review 101(4), 1274-1311.

Glazer, Amihai (1990), ``The Strategy of Candidate Ambiguity," The
American Political Science Review 84(1), 237-241.

Groseclose, Tim (2001), ``A Model of Candidate Location When One
Candidate Has a Valence Advantage," American Journal of Political
Science 45(4), 862-886.

Hinich, Melvin J. and Michael C. Munger (1997), Analytical Politics.
Cambridge: Cambridge University Press.

Kartik, Navin and Preston McAfee, R. P. (2007), “Signaling Character
in Electoral Competition”, American Economic Review 97, 852-870.

Navin Kartik and Richard Van Weelden (2018), ``Informative Cheap
Talk in Elections," The Review of Economic Studies 86 (2),755-784.

Lee, David. S., Enrico Moretti and Butler, M. J. (2004), ``Do Voters
Affect or Elect Policies? Evidence from the U.S. House," Quarterly
Journal of Economics 119(3), 807-859.

Mueller, Dennis (2003), Public Choice III. Cambridge: Cambridge University
Press.

Osborne, Martin and Al Slivinsky (1996), ``A model of Political Competition
with Citizen-Candidates,''\ Quarterly Journal of Economics 111,
64-96.

Panova, Elena (2016), Partially Revealing Campaign Promises. Journal
of Public Economic Theory 19(2), 312-330.

Roemer, John (1997), ``Political Economic Equilibrium When Parties
Represent Constituents: The Unidimensional Case," Social Choice and
Welfare 14, 479-502.

Roemer, John (2001), Political Competition, Cambridge: Harvard University
Press.

Shi, Min, and Svensson, Jakob (2006), ``Political budget cycles:
Do they differ across countries and why?," Journal of Public Economics
90, 1367-1389.

Sulkin, Tracy (2009), ``Campaign Appeals and Legislative Action”,
Journal of Politics 71, 1093--1108.

Vives, Xavier (2001), Oligopoly Pricing: Old Ideas and New Tools,
MIT Press.

Wittman, Donald (1977), ``Candidates with Policy Preferences: A Dynamic
Model," Journal of Economic Theory 14, 180-189. 

\section{Proofs}
\begin{claim}
The function $P\left(\pi_{l},\pi_{r}\right)$ satisfies the following
properties:

Property $(S)$: For every $x_{l},x_{r}\in T$, $P\left(x_{l},x_{r}\right)=1-P\left(x_{r},x_{l}\right)$.

Property $(M)$: For every $x_{l},x_{l}^{\prime},x_{r}\in T$ with
$x_{l}<x_{l}^{\prime}<x_{r}$ $,$ $P\left(x_{l},x_{r}\right)<P\left(x_{l}^{\prime},x_{r}\right)$
and for every $x_{l},x_{l}^{\prime},x_{r}\in T$ with $x_{r}<x_{l}<x_{l}^{\prime}$
$,$ $P\left(x_{l},x_{r}\right)>P\left(x_{l}^{\prime},x_{r}\right).$

Property $(Po)$: For every $x_{l},x_{r}\in T$ $,$ $0<P\left(x_{l},x_{r}\right)<1$. 
\end{claim}

\begin{proof}
First consider property $(S)$. Let $x_{l},x_{r}\in T$. Then 
\[
1-P\left(x_{l},x_{r}\right)=\left\{ \begin{array}{ccc}
1-F\left(\frac{x_{l}+x_{r}}{2}\right) &  & \text{if }x_{l}<x_{r}\\
\frac{1}{2} &  & \text{if \ensuremath{x_{l}=x_{r}}}\\
F\left(\frac{x_{l}+x_{r}}{2}\right) &  & \text{if \ensuremath{x_{l}>x_{r}}}
\end{array}\right.
\]
\[
=P\left(x_{r},x_{l}\right),
\]

which is what we wanted to show.

Now consider Property $(M)$.

Let $x_{l},x_{l}^{\prime},x_{r}\in T$ with $x_{l}<x_{l}^{\prime}<x_{r}$.
Then $P\left(x_{l}^{\prime},x_{r}\right)=F\left(\frac{x_{l}^{\prime}+x_{r}}{2}\right)>F\left(\frac{x_{l}+x_{r}}{2}\right)=P\left(x_{l},x_{r}\right)$,
since $F$ is increasing. Now let $x_{l},x_{l}^{\prime},x_{r}\in T$
with $x_{r}<x_{l}<x_{l}^{\prime}$. Again, since $F$ is increasing,
\[
P\left(x_{l}^{\prime},x_{r}\right)=1-F\left(\frac{x_{l}^{\prime}+x_{r}}{2}\right)<1-F\left(\frac{x_{l}+x_{r}}{2}\right)=P\left(x_{l},x_{r}\right).
\]

Now consider Property $(Po)$. If $x_{l}=x_{r}$, then the result
follows, since $P\left(x_{r},x_{l}\right)=1/2$. Now let $x_{l}\neq x_{r}$.
Then the result follows since $F$ has full support, which means that,
no matter the values of $x_{l}$ and $x_{r}$, there is a positive
probability that $\mathbf{m}$ lies in the interval $\left(0,\frac{x_{l}+x_{r}}{2}\right)$
and in the interval $\left(\frac{x_{l}+x_{r}}{2},1\right)$. 
\end{proof}
\begin{lem}
The best response correspondence \textup{$\varphi_{i}$}
for candidate i with ideal policy $t_{i}$ and platform, $x$, chosen
by i's opponent has the following properties: 
\[
\left\{ \begin{array}{ccc}
\varphi_{t_{i}}\left(x\right) \subset (x, t_{i}] &  & \text{if \ensuremath{x<t_{i}}}\\
\varphi_{t_{i}}\left(x\right)=\{t_{i}\} &  & \text{if \ensuremath{x=t_{i}}}\\
\varphi_{t_{i}}\left(x\right) \subset [t_{i}, x)  &  & \text{if }x>t_{i}
\end{array}\right.
\]
\end{lem}

\begin{proof}
After rearranging terms and eliminating constants that do not depend
on candidate $l$'s choice, the candidate's decision problem simplifies
to:

\[
\underset{x_{l}}{\max}\hspace{0.1cm}P\left(x_{l},x_{r}\right)\left(u\left(x_{l},t_{l}\right)-u\left(x_{r},t_{l}\right)\right).
\]

This function is discontinuous at $x_{l}=x_{r}$, except when $F(x_{r})=0.5$.

Let $x<t_{l}$.

First, notice that, for every $x_{l}\in\left[0,x\right)$, $P\left(x_{l},x\right)\left(u\left(x_{l},t_{l}\right)-u\left(x,t_{l}\right)\right)<0=P\left(x,x\right)\left(u\left(x,t_{l}\right)-u\left(x,t_{l}\right)\right)$.
This shows that $\varphi_{t_{l}}\left(x\right)\subset [x,1]$. Next, notice
that for any $x_{l}\in\left(x,t_{l}\right]$, we have that 
\[
P\left(x_{l},x\right)\left(u\left(x_{l},t_{l}\right)-u\left(x,t_{l}\right)\right)>0=P\left(x,x\right)\left(u\left(x,t_{l}\right)-u\left(x,t_{l}\right)\right).
\]
This shows that $\varphi_{t_{l}}\left(x\right)\subset (x,1]$. Now notice that,
for every $x_{l}\in\left(t_{l},1\right]$,
\[
P\left(t_{l},x\right)\left(u\left(t_{l},t_{l}\right)-u\left(x,t_{l}\right)\right)>P\left(x_{l},x\right)\left(u\left(x_{l},t_{l}\right)-u\left(x,t_{l}\right)\right).
\]
This is so because, by Property $(M)$, $P\left(t_{l},x\right)>P\left(x_{l},x\right)$
and $u\left(t_{l},t_{l}\right)>u\left(x_{l},t_{l}\right)$. This shows
that $\varphi_{t_{l}}\left(x\right)\subset [0,t_{l}]$. We have thus shown
that if $x<t_{l}$ then $\varphi_{t_{l}}\left(x\right)\subset (x,t_{l}]$.

Let $x=t_{l}$.

Then clearly $\varphi_{t_{l}}\left(x\right)=\{t_{l}\}$.

Now let $x>t_{l}$.

As before, notice that, for every $x_{l}\in\left(x,1\right]$, 
\[
P\left(x_{l},x\right)\left(u\left(x_{l},t_{l}\right)-u\left(x,t_{l}\right)\right)<0=P\left(x,x\right)\left(u\left(x,t_{l}\right)-u\left(x,t_{l}\right)\right).
\]
This shows that $\varphi_{t_{l}}\left(x\right)\subset [0,x]$. Next, notice
that for any $x_{l}\in\left[t_{l},x\right)$ we have that $P\left(x_{l},x\right)\left(u\left(x_{l},t_{l}\right)-u\left(x,t_{l}\right)\right)>0=P\left(x,x\right)\left(u\left(x,t_{l}\right)-u\left(x,t_{l}\right)\right)$.
This shows that $\varphi_{t_{l}}\left(x\right)\subset [0,x)$. Now notice that,
for every $x_{l}\in\left[0,t_{l}\right)$, 
\[
P\left(t_{l},x\right)\left(u\left(t_{l},t_{l}\right)-u\left(x,t_{l}\right)\right)>P\left(x_{l},x\right)\left(u\left(x_{l},t_{l}\right)-u\left(x,t_{l}\right)\right).
\]
This is so because, by Property $(M)$, $P\left(t_{l},x\right)>P\left(x_{l},x\right)$
and $u\left(t_{l},t_{l}\right)>u\left(x_{l},t_{l}\right)$. This shows
that $\varphi_{t_{l}}\left(x\right)\subset [t_{l},1]$. We have thus shown
that if $x>t_{l}$ then $\varphi_{t_{l}}\left(x\right)\subset [t_{l},x)$. This completes the proof for $\varphi_{t_{l}}$.

The proof for $\varphi_{t_{r}}$ is similar and we omit it here. 
\end{proof}
\begin{prop}
Assume that SSCP holds. Then if $t_{l}\leq x_{r}<x_{r}^{\prime}\leq t_{r}$,
then we have that $\underline{\varphi}_{t_{l}}\left(x_{r}^{\prime}\right)\geq\overline{\varphi}_{t_{l}}\left(x_{r}\right),$
and if $t_{r}\geq x_{l}^{\prime}>x_{l}\geq t_{l}$, then we have that $\underline{\varphi}_{t_{r}}\left(x_{l}^{\prime}\right)\geq\overline{\varphi}_{t_{r}}\left(x_{l}\right).$ 
\end{prop}
\begin{proof}
Let $t_{l}\leq x_{r}<x_{r}^{\prime}\leq t_{r}$. Let $x_{l}=\overline{\varphi}_{t_{l}}\left(x_{r}\right)\,$ and $x_{l}^{\prime}=\underline{\varphi}_{t_{l}}\left(x_{r}^{\prime}\right).$
By definition of $\overline{\varphi}_{t_{l}}$, 
\[
U_{t_{l}}\left(x_{l},x_{r}\right)\geq U_{t_{l}}\left(x_{l}^{\prime},x_{r}\right).
\]
We want to show that $x_{l}^{\prime}\geq x_{l}.$

If $x_{r}=t_{l}$ then Lemma \ref{G1} implies that $x_{l}=t_{l}$ and $x_{l}^{\prime}\in \left[t_{l},x_{r}^{\prime}\right)$, and hence $x_{l}^{\prime}\geq x_{l}.$

If $x_{r}>t_{l},$ then $x_{l}^{\prime}\geq x_{l}$ follows because,  if $x_{l}>x_{l}^{\prime}$, Lemma \ref{G1} implies that

\[t_{l}\leq x_{l}^{\prime}<x_{l}<x_{r}<x_{r}^{\prime}\leq t_{r}\]

Then, by $SSCP$,
\[
U_{t_{l}}\left(x_{l},x_{r}^{\prime}\right)>U_{t_{l}}\left(x_{l}^{\prime},x_{r}^{\prime}\right)
\]
which contradicts the fact that $x_{l}^{\prime}\,$ is optimal given
$x_{r}^{\prime}$\ for a candidate with ideal policy $t_{l}.$  Hence, $x_{l}^{\prime}\geq x_{l}.$

Combining these implications, we obtain that
$\underline{\varphi}_{t_{l}}\left(x_{r}^{\prime}\right)\geq\overline{\varphi}_{t_{l}}\left(x_{r}\right)$
when $t_{r}\geq x_{r}^{\prime}>x_{r}\geq t_{l}$.

Now let $t_{l}\leq x_{l}<x_{l}^{\prime}\leq t_{r}$. Let $x_{r}=\overline{\varphi}_{t_{r}}\left(x_{l}\right)\,$ and $x_{r}^{\prime}=\underline{\varphi}_{t_{r}}\left(x_{l}^{\prime}\right).$
By definition of $\overline{\varphi}_{t_{r}}$, 
\[
U_{t_{r}}\left(x_{l}^{\prime},x_{r}^{\prime}\right)\geq U_{t_{r}}\left(x_{l}^{\prime},x_{r}\right).
\]
We want to show that $x_{r}^{\prime}\geq x_{r}.$
If $x_{l}^{\prime}=t_{r}$ then Lemma \ref{G1} implies that $x_{r}^{\prime}=t_{r}$ and also that $x_{r}\in \left(x_{l},t_{r}\right]$, and hence $x_{r}\leq x_{r}^{\prime}.$
If $x_{l}^{\prime}<t_{r},$ then $x_{r}^{\prime}\geq x_{r}$ follows because,  if $x_{r}>x_{r}^{\prime}$, Lemma \ref{G1} implies that

\[t_{l}\leq x_{l}<x_{l}^{\prime}<x_{r}^{\prime}<x_{r}\leq t_{r}\]

Then, by $SSCP$,
\[
U_{t_{r}}\left(x_{l},x_{r}^{\prime}\right)>U_{t_{r}}\left(x_{l},x_{r}\right)
\]
which contradicts the fact that $x_{r}$ is optimal given
$x_{l}$ for a candidate with ideal policy $t_{r}.$ Hence, $x_{r}^{\prime}\geq x_{r}.$ 

Combining these implications, we obtain that
$\underline{\varphi}_{t_{r}}\left(x_{l}^{\prime}\right)\geq\overline{\varphi}_{t_{r}}\left(x_{l}\right)$
when $t_{r}\geq x_{l}^{\prime}>x_{l}\geq t_{l}$.
\end{proof}
\begin{prop}
Assume that SSCP holds. Then the set $E$
of Nash equilibria is non-empty and it has (coordinatewise) largest
and smallest elements $\left(\overline{x}_{l}^{*},\overline{x}_{r}^{*}\right)$
and $\left(\underline{x}_{l}^{*},\underline{x}_{r}^{*}\right)$. 
\end{prop}

\begin{proof}
From Lemma \ref{G1} we know that that the map $\left(x_{l},x_{r}\right)\mapsto\left[\overline{\varphi}_{t_{l}}\left(x_{r}\right),\overline{\varphi}_{t_{r}}\left(x_{l}\right)\right]$ takes points in $\left[t_{l},t_{r}\right]\times \left[t_{l},t_{r}\right]$ and maps them to $\left[t_{l},t_{r}\right]\times \left[t_{l},t_{r}\right]$, and from Proposition \ref{G4} we know the map
is non-decreasing in $\left[t_{l},t_{r}\right]\times \left[t_{l},t_{r}\right]$. It follows from Tarski's fixed point theorem that
the set 

\[\overline{E}=\left\{ \left(x_{l},x_{r}\right):\left(\overline{\varphi}_{t_{l}}\left(x_{r}\right),\overline{\varphi}_{t_{r}}\left(x_{l}\right)\right)\geq\left(x_{l},x_{r}\right)\right\} 
\]
is non-empty. Since every Nash equilibrium satisfies
\[
\left(\overline{\varphi}_{t_{l}}\left(x_{r}\right),\overline{\varphi}_{t_{r}}\left(x_{l}\right)\right)\geq\left(x_{l},x_{r}\right)),
\]
then the set of Nash equlibria is non-empty and the greatest element $\left(\overline{x}_{l},\overline{x}_{r}\right)$ of $\overline{E}$ is the greatest Nash equilibrium of the game. The proof for the least element $\left(\underline{x}_{l},\underline{x}_{r}\right)$ is similar and we omit it here. 
\end{proof}
\begin{lem}
In every equilibrium $(x_{l}^{\ast},x_{r}^{\ast})$, $t_{l}\leq x_{l}^{\ast}<x_{r}^{\ast}\leq t_{r}.$ 
\end{lem}

\begin{proof}
The proof follows from combining the logical implications of the properties
of the best responses $\varphi_{t_{l}}$ and $\varphi_{t_{r}}$
as identified in Lemma \ref{G1}.

First, notice there is no equilibrium $\left(x_{l},x_{r}\right)$
with $x_{r}\leq t_{l}$ or with $x_{l}<t_{l}$. To see this, notice
that if $x_{r}\leq t_{l}$ then $x_{r}<x_{l}\leq t_{l}$ since $x_{l}\in \varphi_{t_{l}}\left(x_{r}\right)$
but then $x_{l}<x_{r}\leq t_{r}$ since $x_{r}\in\varphi_{t_{r}}\left(x_{l}\right)$.
On the other hand, $x_{l}<t_{l}$ is only a best response for $l$
if $x_{r}<t_{l}$, which we just showed cannot arise in equilibrium.
Therefore, $x_{l}<t_{l}$ also cannot arise in equilibrium.

Similarly, notice there is no equilibrium $\left(x_{l},x_{r}\right)$
with $x_{l}\geq t_{l}$ or with $x_{r}>t_{r}$. To see this, notice
that if $x_{l}\geq t_{r}$ then $t_{r}\leq x_{r}<x_{l}$ since $x_{r}\in\varphi_{t_{r}}\left(x_{l}\right)$
but then $t_{l}\leq x_{l}<x_{r}$ since $x_{l}\in\varphi_{t_{l}}\left(x_{r}\right)$.
On the other hand, $x_{r}>t_{r}$ is only a best response for $r$
if $x_{l}>t_{r}$, which we just showed cannot arise in equilibrium.
Therefore, $x_{r}>t_{r}$ also cannot arise in equilibrium.

Then, in equilibrium, $t_{l}\leq x_{l}<t_{r}$ and $t_{l}<x_{r}\leq t_{r}$.
It then follows that $t_{l}\leq x_{l}<x_{r}$ since $x_{l}\in\varphi_{t_{l}}\left(x_{r}\right)$.

We have thus established that, in equilibrium, $t_{l}\leq x_{l}^{\ast}<x_{r}^{\ast}\leq t_{r}$.
\end{proof}
\begin{claim}
\label{scp}Assume that $SLS$ holds. If $t_{l}<t_{l}^{\prime}\leq x_{l}<x_{l}^{\prime}<x_{r}$
then 
\[
U_{t_{l}}\left(x_{l}^{\prime},x_{r}\right)\geq U_{t_{l}}\left(x_{l},x_{r}\right)\Rightarrow U_{t_{l}^{\prime}}\left(x_{l}^{\prime},x_{r}\right)>U_{t_{l}^{\prime}}\left(x_{l},x_{r}\right),
\]
and if $x_{l}<x_{r}<x_{r}^{\prime}\leq t_{r}<t_{r}^{\prime}$ then
\[
U_{t_{r}}\left(x_{l},x_{r}^{\prime}\right)\geq U_{t_{r}}\left(x_{l},x_{r}\right)\Rightarrow U_{t_{r}^{\prime}}\left(x_{l},x_{r}^{\prime}\right)>U_{t_{r}^{\prime}}\left(x_{l},x_{r}\right).
\]
\end{claim}

\begin{proof}
Let $t_{l}<t_{l}^{\prime}\leq x_{l}<x_{l}^{\prime}<x_{r}$.

Assume that $U_{t_{l}}\left(x_{l}^{\prime},x_{r}\right)\geq U_{t_{l}}\left(x_{l},x_{r}\right)$.
By the definition of $U$, 
\[
P\left(x_{l}^{\prime},x_{r}\right)u\left(x_{l}^{\prime},t_{l}\right)+\left(1-P\left(x_{l}^{\prime},x_{r}\right)\right)u\left(x_{r},t_{l}\right)\geq
\]
\[
P\left(x_{l},x_{r}\right)u\left(x_{l},t_{l}\right)+\left(1-P_{l}\left(x_{l},x_{r}\right)\right)u\left(x_{r},t_{l}\right),
\]
which boils down to 
\[
\frac{P\left(x_{l}^{\prime},x_{r}\right)}{P\left(x_{l},x_{r}\right)}\frac{\left[u\left(x_{l}^{\prime},t_{l}\right)-u\left(x_{r},t_{l}\right)\right]}{\left[u\left(x_{l},t_{l}\right)-u\left(x_{r},t_{l}\right)\right]}\geq1.
\]
By $SLS$, we have that 
\[
\frac{u\left(x_{l}^{\prime},t_{l}^{\prime}\right)-u\left(x_{r},t_{l}^{\prime}\right)}{u\left(x_{l},t_{l}^{\prime}\right)-u\left(x_{r},t_{l}^{\prime}\right)}>\frac{u\left(x_{l}^{\prime},t_{l}\right)-u\left(x_{r},t_{l}\right)}{u\left(x_{l},t_{l}\right)-u\left(x_{r},t_{l}\right)},
\]
therefore 
\[
\frac{P\left(x_{l}^{\prime},x_{r}\right)}{P\left(x_{l},x_{r}\right)}\frac{\left[u\left(x_{l}^{\prime},t_{l}^{\prime}\right)-u\left(x_{r},t_{l}^{\prime}\right)\right]}{\left[u\left(x_{l},t_{l}^{\prime}\right)-u\left(x_{r},t_{l}^{\prime}\right)\right]}>\frac{P\left(x_{l}^{\prime},x_{r}\right)}{P\left(x_{l},x_{r}\right)}\frac{\left[u\left(x_{l}^{\prime},t_{l}\right)-u\left(x_{r},t_{l}\right)\right]}{\left[u\left(x_{l},t_{l}\right)-u\left(x_{r},t_{l}\right)\right]},
\]
which means that 
\[
\frac{P\left(x_{l}^{\prime},x_{r}\right)}{P\left(x_{l},x_{r}\right)}\frac{\left[u\left(x_{l}^{\prime},t_{l}^{\prime}\right)-u\left(x_{r},t_{l}^{\prime}\right)\right]}{\left[u\left(x_{l},t_{l}^{\prime}\right)-u\left(x_{r},t_{l}^{\prime}\right)\right]}>1,
\]
from where it follows that 
\[
U_{t_{l}^{\prime}}\left(x_{l}^{\prime},x_{r}\right)>U_{t_{l}^{\prime}}\left(x_{l},x_{r}\right).
\]

Now let $x_{l}<x_{r}<x_{r}^{\prime}\leq t_{r}<t_{r}^{\prime}$.

Assume that $U_{t_{r}}\left(x_{l},x_{r}^{\prime}\right)\geq U_{t_{r}}\left(x_{l},x_{r}\right)$.
By the definition of $U$, 
\[
P\left(x_{l},x_{r}^{\prime}\right)u\left(x_{l},t_{r}\right)+\left(1-P\left(x_{l},x_{r}^{\prime}\right)\right)u\left(x_{r}^{\prime},t_{r}\right)\geq
\]
\[
P\left(x_{l},x_{r}\right)u\left(x_{l},t_{r}\right)+\left(1-P_{l}\left(x_{l},x_{r}\right)\right)u\left(x_{r},t_{r}\right),
\]
which boils down to 
\[
\frac{\left[1-P\left(x_{l},x_{r}^{\prime}\right)\right]}{\left[1-P\left(x_{l},x_{r}\right)\right]}\frac{u\left(x_{r}^{\prime},t_{r}\right)-u\left(x_{l},t_{r}\right)}{u\left(x_{r},t_{r}\right)-u\left(x_{l},t_{r}\right)}\geq1.
\]
By $SLS$, we have that 
\[
\frac{u\left(x_{r}^{\prime},t_{r}^{\prime}\right)-u\left(x_{l},t_{r}^{\prime}\right)}{u\left(x_{r},t_{r}^{\prime}\right)-u\left(x_{l},t_{r}^{\prime}\right)}>\frac{u\left(x_{r}^{\prime},t_{r}\right)-u\left(x_{l},t_{r}\right)}{u\left(x_{r},t_{r}\right)-u\left(x_{l},t_{r}\right)},
\]

therefore

\[
\frac{\left[1-P\left(x_{l},x_{r}^{\prime}\right)\right]}{\left[1-P\left(x_{l},x_{r}\right)\right]}\frac{u\left(x_{r}^{\prime},t_{r}^{\prime}\right)-u\left(x_{l},t_{r}^{\prime}\right)}{u\left(x_{r},t_{r}^{\prime}\right)-u\left(x_{l},t_{r}^{\prime}\right)}>\frac{\left[1-P\left(x_{l},x_{r}^{\prime}\right)\right]}{\left[1-P\left(x_{l},x_{r}\right)\right]}\frac{u\left(x_{r}^{\prime},t_{r}\right)-u\left(x_{l},t_{r}\right)}{u\left(x_{r},t_{r}\right)-u\left(x_{l},t_{r}\right)},
\]
which means that 
\[
\frac{\left[1-P\left(x_{l},x_{r}^{\prime}\right)\right]}{\left[1-P\left(x_{l},x_{r}\right)\right]}\frac{u\left(x_{r}^{\prime},t_{r}^{\prime}\right)-u\left(x_{l},t_{r}^{\prime}\right)}{u\left(x_{r},t_{r}^{\prime}\right)-u\left(x_{l},t_{r}^{\prime}\right)}>1,
\]
from where it follows that 
\[
U_{t_{r}^{\prime}}\left(x_{l},x_{r}^{\prime}\right)>U_{t_{r}^{\prime}}\left(x_{l},x_{r}\right).
\]
\end{proof}
\begin{thm}
Assume that SSCP and SLS hold. Let $t_{l}<t_{l}^{\prime}<t_{r}<t_{r}^{\prime}$. Then
\begin{itemize}
    \item $\overline{x}_{l}^{*}\left(t_{l}^{\prime},t_{r}\right)\geq \overline{x}_{l}^{*}\left(t_{l},t_{r}\right)$ and $\overline{x}_{r}^{*}\left(t_{l},t_{r}^{\prime}\right)\geq \overline{x}_{r}^{*}\left(t_{l},t_{r}\right)$
    \item $\underline{x}_{l}^{*}\left(t_{l}^{\prime},t_{r}\right)\geq \underline{x}_{l}^{*}\left(t_{l},t_{r}\right)$ and $\underline{x}_{r}^{*}\left(t_{l},t_{r}^{\prime}\right)\geq \underline{x}_{r}^{*}\left(t_{l},t_{r}\right)$
\end{itemize}

\end{thm}

\begin{proof}
Let $t_{l}<t_{l}^{\prime}.$ Fix $x_{r}>t_{l}$, and let  $x_{l}\in\varphi_{t_{l}}\left(x_{r}\right)$. Notice that, by Lemma \ref{G1}, $x_{l}\in \left[t_{l},x_{r}\right)$. Let $x_{l}^{\prime}\in\varphi_{t_{l}^{\prime}}\left(x_{r}\right).$ We want to show that  $x_{l}^{\prime} \geq x_{l}$.

If $t_{l}^{\prime}\geq x_{l}$,
then Lemma \ref{G1} implies that either $x_{l}^{\prime}\in \left[t_{l}^{\prime},x_{r}\right)$ or $x_{l}^{\prime}\in \left(x_{r},t_{l}^{\prime}\right]$, and either way $x_{l}^{\prime}\geq x_{l}$. 

Now assume that $t_{l}^{\prime}\in\left(t_{l},x_{l}\right)$. By definition of $\varphi_{t_{l}}$, 
\[
U_{t_{l}}\left(x_{l},x_{r}\right)\geq U_{t_{l}}\left(x_{l}^{\prime},x_{r}\right).
\]
We argue that $x_{l}^{\prime}\geq x_{l}.$ Assume not, that is, assume that 
$x_{l}^{\prime}<x_{l}$. Then we obtain that $t_{l}<t_{l}^{\prime}\leq x_{l}^{\prime}<x_{l}<x_{r}$
and, by Claim \ref{scp}, 
\[
U_{t_{l}^{\prime}}\left(x_{l},x_{r}\right)>U_{t_{l}^{\prime}}\left(x_{l}^{\prime},x_{r}\right)
\]
which contradicts the fact that $x_{l}^{\prime}\,$ is optimal given
$x_{r}\ $ for a candidate with ideal policy $t_{l}^{\prime}.$ Hence, $x_{l}^{\prime}\geq x_{l}$ (and, in particular, $\overline{\varphi}_{t_{l}^{\prime}}\left(x_{r}\right)\geq \overline{\varphi}_{t_{l}}\left(x_{r}\right)$ and $\underline{\varphi}_{t_{l}^{\prime}}\left(x_{r}\right)\geq \underline{\varphi}_{t_{l}}\left(x_{r}\right)$) when $x_{r}>t_{l}$.

Let $t_{r}<t_{r}^{\prime}$. Fix $x_{l}<t_{r}^{\prime}$ and let $x_{r}^{\prime}\in\varphi_{t_{r}^{\prime}}\left(x_{l}\right)$. By Lemma \ref{G1}, $x_{r}^{\prime}\in \left(x_{l},t_{r}^{\prime}\right]$. Let  $x_{r}\in\varphi_{t_{r}}\left(x_{l}\right).$ We now want to show that $x_{r}^{\prime}\geq x_{r}.$

If $t_{r}\leq x_{r}^{\prime}$,
then Lemma \ref{G1} implies that either $x_{r}\in \left[t_{r},x_{l}\right)$ or $x_{r}\in \left(x_{l},t_{r}\right]$, and either way $x_{r}^{\prime}\geq x_{r}$.

Now assume that $t_{r}\in\left(x_{r}^{\prime},t_{r}^{\prime}\right)$. By definition of $\varphi_{t_{r}^{\prime}}$, 
\[
U_{t_{r}^{\prime}}\left(x_{l},x_{r}^{\prime}\right)\geq U_{t_{r}^{\prime}}\left(x_{l},x_{r}\right)
\]
 
 We argue that $x_{r}^{\prime}\geq x_{r}.$ Assume not, that is, assume that 
$x_{r}^{\prime}<x_{r}$. We then obtain that $x_{l}<x_{r}^{\prime}<x_{r}\leq t_{r}<t_{r}^{\prime}$
and, by Claim \ref{scp}, 
\[
U_{t_{r}}\left(x_{l},x_{r}^{\prime}\right)>U_{t_{r}}\left(x_{l},x_{r}\right)
\]
which contradicts the fact that $x_{r}$ is optimal given
$x_{l}\ $ for a candidate with ideal policy $t_{r}.$ Hence, $x_{r}^{\prime}\geq x_{r}$ (and, in particular, $\overline{\varphi}_{t_{r}^{\prime}}\left(x_{l}\right)\geq \overline{\varphi}_{t_{r}}\left(x_{l}\right)$ and $\underline{\varphi}_{t_{r}^{\prime}}\left(x_{l}\right)\geq \underline{\varphi}_{t_{r}}\left(x_{l}\right)$) for $t_{r}^{\prime}>x_{l}$.

Let $\overline{x}^{*}\left(t_{l},t_{r}\right)=\sup\left\{ \left(x_{l},x_{r}\right):\left(\overline{\varphi}_{t_{l}}\left(x_{r}\right),\overline{\varphi}_{t_{r}}\left(x_{l}\right)\right)\geq\left(x_{l},x_{r}\right)\right\} $
be the largest fixed point of $\left(x_{l},x_{r}\right)\mapsto\left[\overline{\varphi}_{t_{l}}\left(x_{r}\right),\overline{\varphi}_{t_{r}}\left(x_{l}\right)\right]$,
and therefore the largest Nash equilibrium of the game. This equilibrium
exists, by Proposition \ref{G4-1}.

By Lemma \ref{G2},  $t_{l}\leq \overline{x}_{l}^{*}\left(t_{l},t_{r}\right)<\overline{x}_{r}^{*}\left(t_{l},t_{r}\right)\leq t_{r}.$ 

We obtain that $\overline{x}_{r}^{*}\left(t_{l},t_{r}\right)>t_{l}$ and therefore $\overline{x}_{l}^{*}\left(t_{l}^{\prime},t_{r}\right)\geq \overline{x}_{l}^{*}\left(t_{l},t_{r}\right)$. Similarly, we obtain that $t_{r}^{\prime}>\overline{x}_{l}^{*}\left(t_{l},t_{r}\right)$ and therefore  $\overline{x}_{r}^{*}\left(t_{l},t_{r}^{\prime}\right)\geq \overline{x}_{r}^{*}\left(t_{l},t_{r}\right)$ since we just showed that $\overline{\varphi}_{t_{l}^{\prime}}\left(x_{r}\right)\geq \overline{\varphi}_{t_{l}}\left(x_{r}\right)$ when $x_{r}>t_{l}$, and $\overline{\varphi}_{t_{r}^{\prime}}\left(x_{l}\right)\geq \overline{\varphi}_{t_{r}}\left(x_{l}\right)$ when $t_{r}^{\prime}>x_{l}.$ The case of the smallest equilibrium is analogous. 
\end{proof}

\begin{thm}
Let $x_{r}>t_{l}$, \textup{$x_{l}\in\varphi_{t_{l}}\left(x_{r}\right)$
and $x_{l}^{\prime}>x_{l}$.} Then $\pi\left(x_{l}^{\prime},x_{r}\right)>\pi\left(x_{l},x_{r}\right).$
Let $x_{l}<t_{r}$, \textup{$x_{r}\in\varphi_{t_{r}}\left(x_{l}\right)$
and} \textup{$x_{r}^{\prime}<x_{r}$}. Then $\pi\left(x_{l},x_{r}^{\prime}\right)<\pi\left(x_{l},x_{r}\right).$ 
\end{thm}

\begin{proof}
We first establish the result for changes in $x_{l}$ while keeping
$x_{r}$ fixed. Let $x_{r}>t_{l}$ and $x_{l}\in\varphi_{t_{l}}\left(x_{r}\right).$
From Lemma \ref{G1} we know that $t_{l}\leq x_{l}<x_{r}$. If $x_{l}^{\prime}\geq x_{r}$
the result follows immediately since 
\[
P\left(x_{l}^{\prime},x_{r}\right)x_{l}^{\prime}+\left(1-P\left(x_{l}^{\prime},x_{r}\right)\right)x_{r}>P\left(x_{l},x_{r}\right)x_{l}+\left(1-P\left(x_{l},x_{r}\right)\right)x_{r}
\]
regardless of the value of $P\left(x_{l}^{\prime},x_{r}\right)$ and
$P\left(x_{l},x_{r}\right)$, which are always positive.

Now let $x_{l}^{\prime}\in\left(x_{l},x_{r}\right).$ We want to show
that $\pi\left(x_{l}^{\prime},x_{r}\right)>\pi\left(x_{l},x_{r}\right).$

Assume not, that is, assume that 
\[
P\left(x_{l}^{\prime},x_{r}\right)x_{l}^{\prime}+\left(1-P\left(x_{l}^{\prime},x_{r}\right)\right)x_{r}\leq P\left(x_{l},x_{r}\right)x_{l}+\left(1-P\left(x_{l},x_{r}\right)\right)x_{r}.
\]
Cancelling and rearranging terms yields 
\[
\frac{P\left(x_{l}^{\prime},x_{r}\right)}{P\left(x_{l},x_{r}\right)}\geq\frac{x_{r}-x_{l}}{x_{r}-x_{l}^{\prime}}.
\]
Since $x_{l}\in\varphi_{t_{l}}\left(x_{r}\right)$, it follows that
\[
P\left(x_{l},x_{r}\right)u\left(x_{l},t_{l}\right)+\left(1-P\left(x_{l},x_{r}\right)\right)u\left(x_{r},t_{l}\right)\geq
\]
\[
P\left(x_{l}^{\prime},x_{r}\right)u\left(x_{l}^{\prime},t_{l}\right)+\left(1-P\left(x_{l}^{\prime},x_{r}\right)\right)u\left(x_{r},t_{l}\right).
\]
We obtain that 
\[
\frac{P\left(x_{l}^{\prime},x_{r}\right)}{P\left(x_{l},x_{r}\right)}\leq\frac{u\left(x_{r},t_{l}\right)-u\left(x_{l},t_{l}\right)}{u\left(x_{r},t_{l}\right)-u\left(x_{l}^{\prime},t_{l}\right)}.
\]
Putting these expressions together yields 
\[
\frac{u\left(x_{r},t_{l}\right)-u\left(x_{l},t_{l}\right)}{u\left(x_{r},t_{l}\right)-u\left(x_{l}^{\prime},t_{l}\right)}\geq\frac{x_{r}-x_{l}}{x_{r}-x_{l}^{\prime}},
\]
which is equivalent to 
\begin{equation}
u\left(x_{l}^{\prime},t_{l}\right)\leq\frac{x_{r}-x_{l}^{\prime}}{x_{r}-x_{l}}u\left(x_{l},t_{l}\right)+\left(1-\frac{x_{r}-x_{l}^{\prime}}{x_{r}-x_{l}}\right)u\left(x_{r},t_{l}\right).\label{eq:convexo}
\end{equation}
But $\frac{x_{r}-x_{l}^{\prime}}{x_{r}-x_{l}}x_{l}+\left(1-\frac{x_{r}-x_{l}^{\prime}}{x_{r}-x_{l}}\right)x_{r}=x_{l}^{\prime}$
and $\frac{x_{r}-x_{l}^{\prime}}{x_{r}-x_{l}}\in\left(0,1\right)$.
Therefore, strict concavity requires that 
\begin{equation}
u\left(x_{l}^{\prime},t_{l}\right)>\frac{x_{r}-x_{l}^{\prime}}{x_{r}-x_{l}}u\left(x_{l},t_{l}\right)+\left(1-\frac{x_{r}-x_{l}^{\prime}}{x_{r}-x_{l}}\right)u\left(x_{r},t_{l}\right).\label{eq: concavo}
\end{equation}
The contradiction between equations (\ref{eq:convexo}) and (\ref{eq: concavo})
establishes the first result.

Now we establish the result for changes in $x_{r}$ while keeping
$x_{l}$ fixed.

Let $x_{l}<t_{r}$, $x_{r}\in\varphi_{t_{r}}\left(x_{l}\right)$ and
$x_{r}^{\prime}<x_{r}$. From Lemma \ref{G1} we know that $x_{l}<x_{r}\leq t_{r}$.
If $x_{r}^{\prime}\leq x_{l}$ the result follows immediately since
\[P\left(x_{l},x_{r}^{\prime}\right)x_{l}+\left(1-P\left(x_{l},x_{r}^{\prime}\right)\right)x_{r}^{\prime}<P\left(x_{l},x_{r}\right)x_{l}+\left(1-P\left(x_{l},x_{r}\right)\right)x_{r}
\]
regardless of the value of the probabilities $P\left(x_{l},x_{r}^{\prime}\right)$
and $P\left(x_{l},x_{r}\right)$, which are always positive.

Now let $x_{r}^{\prime}\in\left(x_{l},x_{r}\right).$ We want to show
that $\pi\left(x_{l},x_{r}^{\prime}\right)<\pi\left(x_{l},x_{r}\right).$

Assume not, that is, assume that 
\[
P\left(x_{l},x_{r}^{\prime}\right)x_{l}+\left(1-P\left(x_{l},x_{r}^{\prime}\right)\right)x_{r}^{\prime}\geq P\left(x_{l},x_{r}\right)x_{l}+\left(1-P\left(x_{l},x_{r}\right)\right)x_{r}.
\]
Cancelling and rearranging terms yields 
\[
\frac{1-P\left(x_{l},x_{r}^{\prime}\right)}{1-P\left(x_{l},x_{r}\right)}\geq\frac{x_{r}-x_{l}}{x_{r}^{\prime}-x_{l}}.
\]
Since $x_{r}\in\varphi_{t_{r}}\left(x_{l}\right)$, it follows that
\[
P\left(x_{l},x_{r}\right)u\left(x_{l},t_{r}\right)+\left(1-P\left(x_{l},x_{r}\right)\right)u\left(x_{r},t_{r}\right)\geq
\]
\[
P\left(x_{l},x_{r}^{\prime}\right)u\left(x_{l},t_{l}\right)+\left(1-P\left(x_{l},x_{r}^{\prime}\right)\right)u\left(x_{r}^{\prime},t_{l}\right).
\]
We obtain that 
\[
\frac{1-P\left(x_{l},x_{r}^{\prime}\right)}{1-P\left(x_{l},x_{r}\right)}\leq\frac{u\left(x_{r},t_{r}\right)-u\left(x_{l},t_{r}\right)}{u\left(x_{r}^{\prime},t_{r}\right)-u\left(x_{l},t_{r}\right)}.
\]
Putting these expressions together yields 
\[
\frac{u\left(x_{r},t_{r}\right)-u\left(x_{l},t_{r}\right)}{u\left(x_{r}^{\prime},t_{r}\right)-u\left(x_{l},t_{r}\right)}\geq\frac{x_{r}-x_{l}}{x_{r}^{\prime}-x_{l}},
\]
which is equivalent to 
\begin{equation}
u\left(x_{r}^{\prime},t_{r}\right)\leq\frac{x_{r}^{\prime}-x_{l}}{x_{r}-x_{l}}u\left(x_{r},t_{r}\right)+\left(1-\frac{x_{r}^{\prime}-x_{l}}{x_{r}-x_{l}}\right)u\left(x_{l},t_{r}\right).\label{eq:convexo-1}
\end{equation}
But $\frac{x_{r}^{\prime}-x_{l}}{x_{r}-x_{l}}x_{r}^{\prime}+\left(1-\frac{x_{r}^{\prime}-x_{l}}{x_{r}-x_{l}}\right)x_{l}=x_{r}^{\prime}$
and $\frac{x_{r}^{\prime}-x_{l}}{x_{r}-x_{l}}\in\left(0,1\right)$.
Therefore, strict concavity requires that 
\begin{equation}
u\left(x_{r}^{\prime},t_{r}\right)>\frac{x_{r}^{\prime}-x_{l}}{x_{r}-x_{l}}u\left(x_{r},t_{r}\right)+\left(1-\frac{x_{r}^{\prime}-x_{l}}{x_{r}-x_{l}}\right)u\left(x_{l},t_{r}\right).\label{eq: concavo-1}
\end{equation}
The contradiction between equations (\ref{eq:convexo-1}) and (\ref{eq: concavo-1})
establishes the second result. 
\end{proof}

\end{document}